\documentclass{article}
\usepackage[letterpaper, right=1in, left=1in, top=1in, bottom=1in]{geometry}

\usepackage{cite}
\usepackage{amsmath,amssymb,amsfonts}
\usepackage{algorithmic}
\usepackage{graphicx}
\usepackage{textcomp}
\usepackage{xcolor}
\def\BibTeX{{\rm B\kern-.05em{\sc i\kern-.025em b}\kern-.08em
    T\kern-.1667em\lower.7ex\hbox{E}\kern-.125emX}}

\usepackage{tikz}
\usetikzlibrary{shapes,arrows,positioning,calc}
\usepackage{algorithm}
\usepackage{amsthm}
\newtheorem{theorem}{Theorem}

\newcommand{\R}{\mathbb{R}}
\newcommand{\N}{\mathbb{N}}
\newcommand{\bmtx}{\begin{bmatrix}}
\newcommand{\emtx}{\end{bmatrix}}
\newcommand{\bsmtx}{\left[ \begin{smallmatrix}} 
\newcommand{\esmtx}{\end{smallmatrix} \right]}

\tikzstyle{block} = [draw, fill=white, rectangle, 
    minimum height=2em, minimum width=3em]
\tikzstyle{sum} = [draw, fill=white, circle, node distance=1.5cm]
\tikzstyle{input} = [coordinate]
\tikzstyle{output} = [coordinate]
\tikzstyle{pinstyle} = [pin edge={to-,thin,black}]

\begin{document}

\title{\LARGE \bf Robust Online Convex Optimization for Disturbance Rejection}

\author{Joyce Lai and Peter Seiler
	\thanks{J. Lai and P. Seiler are with the Department of Electrical Engineering and Computer Science at the University of Michigan, Ann Arbor, MI 48109, USA. {\tt\small \{joycelai,pseiler\}@umich.edu}
    }%
}


\maketitle

\begin{abstract}
Online convex optimization (OCO) is a powerful tool for learning sequential data, making it ideal for high precision control applications where the disturbances are arbitrary and unknown in advance. However, the ability of OCO-based controllers to accurately learn the disturbance while maintaining closed-loop stability relies on having an accurate model of the plant. This paper studies the performance of OCO-based controllers for linear time-invariant (LTI) systems subject to disturbance and model uncertainty. The model uncertainty can cause the closed-loop to become unstable. We provide a sufficient condition for robust stability based on the small gain theorem. This condition is easily incorporated as an on-line constraint in the OCO controller. Finally, we verify via numerical simulations that imposing the robust stability condition on the OCO controller ensures closed-loop stability.
\end{abstract}


\section{Introduction}
\label{sec:intro}

This paper considers a class of controllers recently developed using online convex optimization (OCO).  Online machine learning and convex optimization methods are powerful tools for learning sequential data. This makes these techniques ideal for high precision control applications like satellite pointing and photolithography. These systems have reliable physics-based models with small error (within the control bandwidth) but are subject to unknown arbitrary disturbances.

This has motivated a large body of recent work using online learning and convex optimization for control
\cite{anava14OCO,hazan11OML,zinkevich03OCP,hazan07,shalev11,hazan16,agarwal19ANIPS,foster20,goel22arXiv}.
The most closely related work is the class of OCO controllers defined in \cite{agarwal19}. Here, OCO with memory is introduced to the discrete-time control setting as an ideal cost minimization problem (which we describe in detail in Section~\ref{sec:OPGD}) to handle arbitrary disturbances and general time-varying convex cost functions. The OCO controller has promising regret guarantees and makes less restrictive assumptions about the disturbance characteristics (e.g., white noise or worst-case) than that of $H_2$ and $H_\infty$ optimal control techniques \cite{zhou95,skogestad05}. This makes OCO methods well suited for high precision control applications with unknown, arbitrary disturbances that degrade the system performance.

The OCO framework in \cite{agarwal19} aims to learn the disturbance characteristics in real time. However, small model errors can cause instability and thus must be explicitly considered in the design. There are additional works that attempt to learn the model from data \cite{rahman2016tutorial,venugopal2000adaptive, santillo2010adaptive,goel22arXivGH,goel21arXiv,goel21PMLR,goel20arXiv}. However, dynamic uncertainties in many high precision applications arise due to high frequency, time-varying, and/or nonlinear effects. It is difficult to learn such unmodeled effects from real-time data. In these cases, it is useful to design a robust OCO-based controller that can learn the disturbance features and tolerate model uncertainty, thus motivating our work.

There are three main contributions of our work. First, we provide a robust stability condition for OCO  control of a discrete linear time-invariant (LTI) plant (Theorem~\ref{thm:scaledSG}
 in Section~\ref{sec:scaledSG}). The scaled small gain condition is written abstractly with an arbitrary choice of an induced system norm.  Our second contribution is to present  a constrained OCO (C-OCO) control algorithm which is robust to nonparametric model uncertainties (Section~\ref{sec:application}). This algorithm uses a specific implementation of the scaled small gain condition with the induced $\ell_\infty$-norm (Section~\ref{sec:boundingMLTV}). This particular choice for the induced norm enables easy implementation of the robust stability condition in the C-OCO algorithm.  The third contribution is to  present numerical results that illustrate the effect of this robust stability constraint on the OCO controller (Section~\ref{sec:numerical}).


\section{Problem Formulation}
\label{sec:probform}

This section formulates the OCO control problem for discrete-time LTI plants subject to both model uncertainty and unknown disturbances.

\subsection{Notation}

Let $v \in \R^{n}$ be a vector. The $p$-norm of this vector is defined as $\|v\|_p := \big[ \sum_{i=1}^n |v_i|^p \big]^{\frac{1}{p}}$.
Next, $\N$ denotes the set of non-negative integers. Let $d:\N \to \R^{n}$ denote a vector-valued sequence $\{d_0,d_1,\ldots\}$. 
The $\ell_p$-norm of $d$ is defined as:
\begin{align}
\|d\|_p = \left[ \sum_{t=0}^\infty \|d_t\|_p^p \right]^{\frac{1}{p}}.
\end{align}
Note that $\|d_t\|_p$ is the $p$-norm of the vector $d_t\in \R^n$ at time $t$ while $\|d\|_p$ is the $\ell_p$-norm of the sequence. The set $\ell_p$ consists of sequences that have finite $\ell_p$-norm. The subset $\ell_{pe} \subset \ell_p$ is the extended space of sequences that have finite $\ell_p$-norm on all finite intervals, i.e. 
$\sum_{t=0}^T \|d_t\|_p^p < \infty$ for all $T\ge 0$. Finally, let $G: \ell_{pe} \to \ell_{pe}$ denote systems that map an input signal $u\in \ell_{pe}$ to an output signal $y\in \ell_{pe}$.  The induced $\ell_p$-norm for this system is defined as:
\begin{align}
    \|G\|_{p\to p} = \sup_{0\ne u \in \ell_{p}} \frac{\|y\|_p}{\|u\|_p}.
\end{align}
To simplify notation, we'll often use $\|d\|$ and $\|G\|$ for the signal norm and system induced norm when the specific $p$-norm is not important.

\subsection{Model Uncertainty}
\label{sec:model-uncertainty}

In this section, we consider the feedback system in Figure~\ref{fig:uncertainty} and discuss the model uncertainty $\Delta(z)$ in more detail.


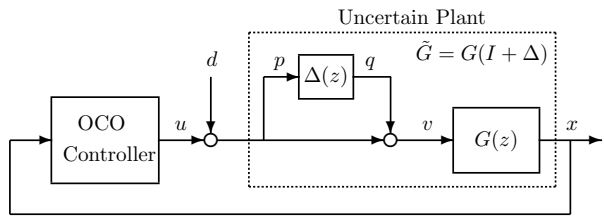
\begin{figure}[h!]
\centering
\scalebox{0.8}{
\begin{picture}(280,105)(20,-45)
 \thicklines 
 \put(20,0){\vector(1,0){20}}  
 \put(40,-20){\framebox(50,40){}}
 \put(52,5){OCO}
 \put(45,-10){Controller}
 \put(90,0){\vector(1,0){22}}  
 \put(98,5){$u$}
 \put(115,0){\circle{6}}
 \put(115,30){\vector(0,-1){27}}  
 \put(113,35){$d$}
 \put(118,0){\vector(1,0){79}}  
 \put(200,0){\circle{6}}  
 \put(140,0){\line(0,1){30}}  
 \put(140,30){\vector(1,0){17}}
 \put(145,35){$p$}
 \put(157,20){\framebox(25,20){$\Delta(z)$}}
 \put(188,35){$q$}
 \put(212,38){$\tilde{G}=G (I+\Delta)$}
 \put(175,55){Uncertain Plant}
 \put(133,-22){\dashbox(145,73)}
 \put(215,5){$v$}
 \put(182,30){\line(1,0){18.25}}  
 \put(200,30){\vector(0,-1){27}}  
 \put(203,0){\vector(1,0){27}}  
 \put(230,-15){\framebox(40,30){$G(z)$}}
 \put(270,0){\vector(1,0){30}}  
 \put(283,5){$x$}
 \put(285,0){\line(0,-1){35}}  
 \put(285,-35){\line(-1,0){265}}  
 \put(20,-35){\line(0,1){35}}  
\end{picture}
}
\caption{Discrete-time feedback system with unknown disturbance $d$ and uncertainty $\Delta(z)$. OCO control is used to reject  the disturbance $d$ without knowledge of the uncertainty $\Delta(z)$.}
 \label{fig:uncertainty}
\end{figure}

Consider the nominal discrete-time, LTI plant $G(z)$ with dynamics:
\begin{align}
\label{eq:nominal-plant}
	x_{t+1} = A\,x_t + B\,v_t,
\end{align}
where $x_t\in\R^{n_x}$ and $v_t\in\R^{n_u}$ are the nominal plant state and input at time $t$, respectively. We assume $x_0=0$ for simplicity.

Model uncertainty for systems with  physics-based models often shows up as unmodeled actuator dynamics affecting the plant input \cite{zhou95, doyle78, skogestad05}. We can account for these unmodeled dynamics by defining an input-multiplicative uncertainty set $\mathcal{G}_{\delta}$ as:
\begin{align}
\label{eq:uncertainty-set}
	\mathcal{G}_\delta
	=
	\left\{
	\tilde{G}(z) = G(z)(I + \Delta(z))
	:
	\Vert\Delta\Vert\le\delta
	\right\},
\end{align}
where $\delta\in[0,\infty)$. Note that the induced $2$-norm is common choice to bound the uncertainty. However, our main result in Section~\ref{sec:main} holds for any induced $p$-norm.

Let $\tilde{G}_0(z)$ denote the true plant dynamics. We assume that the true plant is within the uncertainty set, i.e. $\tilde{G}_0(z) \in \mathcal{G}_\delta$. In other words, there exists a specific $\Delta_0(z)$ such that $\|\Delta_0\| \le \delta$ and $\tilde{G}_0(z) = G(z)(I+\Delta_0(z))\in\mathcal{G}_\delta$.
More generally, we refer to $\tilde{G}(z)=G(z)(I+\Delta(z))$ as the uncertain plant. An alternative viewpoint is that the uncertain plant is $\tilde{G}(z)=G(z)F(z)$ where $F(z)=I+\Delta(z)$ represents unmodeled dynamics. Note that we assume the uncertainty $\Delta(z)$ is LTI.  However, our main result in Section~\ref{sec:main} can be extended to the case where $\Delta$ is a possibly nonlinear time-varying (NLTV) system.

\subsection{OCO Control}
\label{sec:oco-control}

This section describes the OCO controller. We consider the feedback system in Figure~\ref{fig:oco-control} where the OCO controller is shown in more detail.

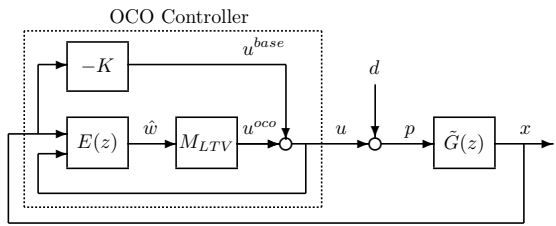
\begin{figure}[h!]
\centering
\scalebox{0.75}{
\begin{picture}(275,110)(20,-45)
 \thicklines 
 \put(20,0){\vector(1,0){30}}  
 \put(50,-17){\framebox(30,25){$E(z)$}}
 \put(35,35){\vector(1,0){15}}  
 \put(35,0){\line(0,1){35}}  
 \put(80,-5){\vector(1,0){25}}  
 \put(135,-5){\vector(1,0){22}}  
 \put(88,0){$\hat{w}$}
 \put(105,-17){\framebox(30,25){$M_{LTV}$}}
 \put(138,0){$u^{oco}$}
 \put(135,-5){\vector(1,0){22}}  
 \put(50,22){\framebox(30,24){$-K$}}
 \put(80,35){\line(1,0){80}}  
 \put(138,40){$u^{base}$}
 \put(160,35){\vector(0,-1){38}}  
 \put(170,-5){\line(0,-1){25}}  
 \put(170,-30){\line(-1,0){135}}  
 \put(35,-30){\line(0,1){20}}  
 \put(35,-10){\vector(1,0){15}}  
 \put(28,-37){\dashbox(150,89)}
 \put(71,56){OCO Controller}
 \put(160,-5){\circle{6}}
 \put(163,-5){\vector(1,0){39}}  
 \put(185,0){$u$}
 \put(205,-5){\circle{6}}
 \put(202,30){$d$}
 \put(205,25){\vector(0,-1){27}}  
 \put(220,0){$p$}
 \put(208,-5){\vector(1,0){27}}  
 \put(235,-17){\framebox(30,25){$\tilde{G}(z)$}}
 \put(265,-5){\vector(1,0){30}}  
 \put(278,0){$x$}
 \put(280,-5){\line(0,-1){40}}  
 \put(280,-45){\line(-1,0){260}}  
 \put(20,-45){\line(0,1){45}}  
\end{picture}
}
\caption{Block diagram representation of the OCO controller in a discrete-time feedback system with unknown disturbance $d_t$ and uncertain plant $\tilde{G}(z)$. The OCO controller is composed of a state-feedback gain $K$,  an estimator $E(z)$, and an LTV system $M_{LTV}$.}
\label{fig:oco-control}
\end{figure}

Unknown disturbances are often caused by environmental factors and moving physical components which degrade system performance. However, these disturbances often also have learnable characteristics. It is typical to model such disturbances as entering at the plant input as shown in Figure~\ref{fig:oco-control}.

OCO control can be used to learn and reject the disturbance without a priori knowledge of the disturbance \cite{anava14OCO,hazan11OML,zinkevich03OCP,hazan07,shalev11,hazan16,agarwal19ANIPS,foster20,goel22arXiv}. Here, we describe a class of OCO controllers closely related to \cite{agarwal19} which considers the case when $\Delta(z)=0$. The OCO controller has the block diagram representation shown in Figure~\ref{fig:oco-control}.
This corresponds to the class of disturbance action controllers defined as:
\begin{align}
\label{eq:oco-control}
	u_t = -Kx_t + \sum_{i=0}^{H-1} M_t^{[i]} \hat{w}_{t-i},
\end{align}
where $K\in\R^{n_u \times n_x}$, $M_t^{[i]}\in\R^{n_u \times n_x}$, and $\hat{w}_t\in\R^{n_x}$ are the state-feedback gain, learned coefficients, and disturbance estimate, at time $t$, respectively. The state-feedback gain $K$ is user-selected while the learned coefficients $\{M_t\}_{i=0}^{H-1}$are typically updated via some online optimization method. For example, \cite{agarwal19} uses online projected gradient descent (OPGD) with memory (see Section~\ref{sec:OPGD}).

The disturbance estimate $\hat{w}_t$ is assumed to be the output of an LTI estimator $E(z)$ with dynamics:
\begin{align}
\label{eq:estimator}
\begin{split}    
	x_{t+1}^e &= A_e x_t^e + B_{e1} x_t + B_{e2} u_t \\
	\hat{w}_t &= C_e x_t^e + D_{e1} x_t + D_{e2} u_t,
\end{split}
\end{align}
where $x_t^e\in\R^{n_e}$ and $\hat{w}_t\in\R^{n_x}$ are the estimator state and output at time $t$, respectively. Typically, $\hat{w}_t$ is an estimate of $Bd_t \in \R^{n_x}$ (possibly with delay), i.e., it is an estimate the disturbance effect on the (nominal) state. The estimate is constructed from $x_t$ and $u_t$. This estimator is motivated by the case when $\Delta(z)=0$.

The first term in (\ref{eq:oco-control}) is considered the baseline controller which we denote by:
\begin{align}
\label{eq:baseline}
	u_t^{base}=-Kx_t.
\end{align}
The main results in Section~\ref{sec:main} can be generalized to the case when the baseline control $u_t^{base}$ is the output of an LTI controller $K(z)$ with input $x_t$. We assume the baseline controller is a static, state-feedback gain for simplicity.

The second term in (\ref{eq:oco-control}) is the output of an finite impulse response (FIR) filter with time-varying coefficients. We denote the FIR filter with time-varying coefficients as a linear time-varying (LTV) system $M_{LTV}$ with input-output dynamics defined as:
\begin{align}
\label{eq:mltv}
	u_t^{oco}=\sum_{i=0}^{H-1} M_t^{[i]} \hat{w}_{t-i}.
\end{align}
where $\hat{w}_t\in\R^{n_x}$ and $u_t^{oco}\in\R^{n_u}$ are the input and output at time $t$, respectively. The FIR filter order $H$ is also referred to as the learning horizon since the coefficients are often updated via OCO using the past $H$ disturbance estimates. We provide an example of online optimization in Sections~\ref{sec:application} and~\ref{sec:numerical}, but the main results in Section~\ref{sec:main} assume only that the coefficients are time-varying.

Given (\ref{eq:baseline}) and (\ref{eq:mltv}), the OCO controller (\ref{eq:oco-control}) can be interpreted as a baseline controller $u_t^{base}$ plus an adapting term $u_t^{oco}$ which corrects for the unknown disturbance $d_t$ based on disturbance estimates.

\subsection{Model Uncertainty Effects on OCO Control}
\label{sec:uncertainty-effects}

The uncertainty $\Delta(z)$ and disturbance $d_t$ have different effects on closed-loop stability. Suppose the state-feedback gain $K$ is stabilizing, i.e., all eigenvalues of $(A-BK)$ are strictly inside the unit disk. Given a perfect plant model, i.e., $\Delta(z)=0$, OCO control can be designed to achieve disturbance rejection with provable guarantees \cite{agarwal19}. In this case, a bounded disturbance $d$ cannot cause signals $x,u,\hat{w}$, etc. to grow unbounded. However, small amounts of model uncertainty can cause the system to become unstable.

As shown in Figures~\ref{fig:uncertainty} and \ref{fig:oco-control}, the (true) plant input is the control input perturbed by an unknown disturbance:
\begin{align}
	p_t = u_t+d_t,
\end{align}
where $u_t,d_t,p_t\in\R^{n_u}$ are the control input, disturbance, and perturbed (true) plant input at time $t$, respectively. The perturbed input $p_t$ is further distorted by the uncertainty $\Delta(z)$. The resulting input to the nominal plant $G(z)$ is:
\begin{align}
v_t &= (I+\Delta)\,p_t = u_t + d_t + q_t,
\end{align}
where $q_t=\Delta p_t\in\R^{n_u}$. Again, $v_t$ is the nominal plant input at time $t$. Not only is there an unknown disturbance $d_t$, but also a distorted signal $q_t$ due to uncertainty $\Delta(z)$.

The additional perturbation $q_t$ can lead to unexpected behaviors that affect the disturbance estimate and FIR filter coefficient update when left unaccounted for in the OCO design. This can occur even when the state-feedback gain $K$ is stabilizing for the true plant $\tilde{G}(z)$. Thus, the OCO controller is required to: i) learn and compensate for the disturbance, and ii) stabilize the system in the presence of uncertainty. The OCO controller must achieve these objectives without a priori knowledge of the disturbance or uncertainty.

\section{Main Result}
\label{sec:main}

This section provides a condition on $M_{LTV}$ that ensures the feedback system with OCO control remains stable even in the presence of the model uncertainty.

\subsection{Linear Fractional Transformation}
\label{sec:lft}

As a first step, we transform the feedback system of the OCO controller and uncertain plant (Figures~\ref{fig:uncertainty} and~\ref{fig:oco-control}) to a standard form as shown in Figure~\ref{fig:LFTdiagram}. This diagram separates the LTI dynamics $P$ from the uncertainty $\Delta$ and time-varying OCO dynamics $M_{LTV}$. Here $P$ includes the  dynamics due to the plant, estimator, and state-feedback gain. This diagram is called a linear fractional transformation (LFT) in the robust control literature \cite{zhou95,skogestad05}. We use the notation $F_U(P,\Gamma)$ for this interconnection with $\Gamma=\bsmtx \Delta & 0 \\ 0 & M_{LTV} \esmtx$ closed around the upper channels of $P$. \newline

\begin{figure}[h!t]
\centering
\begin{picture}(110,90)(40,20)
 \thicklines
 \put(75,25){\framebox(40,40){$P$}}
 \put(143,40){$d$}
 \put(150,35){\vector(-1,0){35}}  
 \put(42,40){$x$}
 \put(75,35){\vector(-1,0){35}}  
\put(75,73){\framebox(40,40){}}
 \put(79,100){{\footnotesize $\Delta$}}
 \put(98,100){{\footnotesize $0$}}
 \put(80,83){{\footnotesize $0$}}
 \put(91,83){{\scriptsize $M_{LTV}$}}
 \put(118,105){$\Gamma$}
 \put(34,70){$\bmtx p \\ \hat{w} \emtx$}
 \put(55,55){\line(1,0){20}}  
 \put(55,55){\line(0,1){35}}  
 \put(55,90){\vector(1,0){20}}  
 \put(140,70){$\bmtx q \\ u^{oco} \emtx$}
 \put(135,90){\line(-1,0){20}}  
 \put(135,55){\line(0,1){35}}  
 \put(135,55){\vector(-1,0){20}}  
 \end{picture}
\caption{Equivalent LFT $F_U(P,\Gamma)$ of original system separating LTI dynamics $P$ from uncertainty $\Delta$ and time-varying learning dynamics $M_{LTV}$.}
\label{fig:LFTdiagram}
\end{figure}
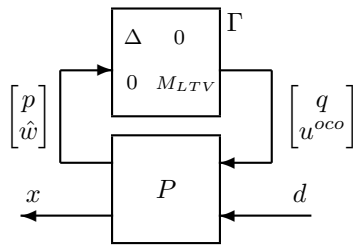

An explicit state-space model for $P$ can be determined from the various components of the feedback system described in Section~\ref{sec:probform}.  The dynamics of $P$ are given by:
\begin{align*}
\bmtx x_{t+1} \\ x_{t+1}^e \emtx & =
\bmtx A-BK & 0 \\ B_{e1}-B_{e2}K & A_e \emtx
\bmtx x_t \\ x_t^e \emtx
+ \bmtx B & B & B \\ 0 & B_{e2} & 0 \emtx
\bmtx q_t \\ u_t^{oco} \\ d_t \emtx \\
\bmtx p_t \\ \hat{w}_t \\ x_t \emtx & = 
\bmtx -K & 0 \\ D_{e1}-D_{e2}K & C_e  \\ I & 0\emtx
\bmtx x_t \\ x_t^e \emtx
+ \bmtx 0 & I & I \\ 0 & D_{e2} & 0 \\ 0 & 0 & 0\emtx
\bmtx q_t \\ u_t^{oco} \\ d_t \emtx
\end{align*}
We use the LFT representation $F_U(P,\Gamma)$ to formulate and state our robust stability theorem in the next subsection.

\subsection{Scaled Small Gain Theorem}
\label{sec:scaledSG}

Our first stability result is a variation of the standard small gain theorem (see Section 5.4 of \cite{khalil02}).  This provides a sufficient condition for the dynamics $F_U(P,\Gamma)$ to have a bounded gain from disturbance $d$ to state $x$. Note stability here is in the sense of bounded gain in some induced norm.
 
\begin{theorem}
\label{thm:smallgain}
 Consider the interconnection $F_U(P,\Gamma)$ 
where $P:\ell_{pe} \to \ell_{pe}$ and $\Gamma:\ell_{pe} \to \ell_{pe}$ are linear systems with finite induced $\ell_p$-norm. Partition $P$ as:
\begin{align}
\bmtx \bar{p} \\ x \emtx
=
\bmtx P_{11} & P_{12} \\
P_{21} & P_{21} \emtx
\,
\bmtx \bar{q} \\ d \emtx,
\end{align}
where $\bar{p}:= \bsmtx p \\ \hat{w} \esmtx$ and  $\bar{q}:= \bsmtx q \\ u^{oco} \esmtx$ are the inputs and outputs of $\Gamma$.
The interconnection has finite induced $\ell_p$-norm, i.e.
$\|F_U(P,\Gamma)\| < \infty$,
if $\|P_{11} \|\, \|\Gamma\| <1$.
\end{theorem}
\begin{proof}
The system $P$ is LTI so by the principle of superposition (assuming zero initial conditions):
\begin{align}
    \bar{p} = P_{11} \bar{q} + P_{12} d.
\end{align}
We can bound $\bar{p}$ using the triangle inequality and the definition of the induced norm:
\begin{align}
\label{eq:pbarBound}
   \|\bar{p}\| \le  
    \|P_{11}\| \, \|\bar{q}\|+ \|P_{12} \| \, \| d\|.
\end{align}
Next, $\bar{q} = \Gamma \bar{p}$
so that $\|\bar{q}\| \le \|\Gamma\|\,  \| \bar{p} \|$. Substitute this bound into \eqref{eq:pbarBound} and re-arrange to obtain:
\begin{align}
    \|\bar{p}\|\le 
      \frac{\|P_{12}\|}{1-\|P_{11}\| \|\Gamma \|}  \, \|d\|.
\end{align}
This last step requires the small gain condition $\|P_{11} \| \, \|\Gamma\| <1$ to obtain the bound on $\|\bar{p}\|$.

Finally, the state is $x=P_{21} \bar{q} + P_{22} d$. We can use similar steps and the bound on $\bar{p}$ to obtain:
\begin{align}
    \|x\| \le
    \left[ 
\|P_{22}\|
+\frac{\|P_{21}\|\,
\|P_{12} \| \, \| \Gamma\|}{1-\|P_{11}\| \, \|\Gamma \|}
    \right] \|d\|.
\end{align}
Hence, $F_U(P,\Gamma)$ has finite induced $\ell_p$-norm.   
\end{proof}

The small gain condition in the previous theorem can be conservative as it does not exploit the block structure $\Gamma = \bsmtx \Delta & 0 \\ 0 & M_{LTV} \esmtx$. We can reduce the conservatism by normalizing the blocks and introducing scalings.  Specifically, assume $\|\Delta \|\le \delta$ and $\|M_{LTV}\| \le \beta$. Define the normalized uncertainty and learning dynamics as:
$\tilde{\Delta} = \frac{1}{\delta} \Delta$ and 
$\tilde{M}_{LTV} = \frac{1}{\beta} M_{LTV}$. Stacking these together yields
\begin{align}
\label{eq:GamTil}
\tilde{\Gamma} := 
   \bmtx \frac{1}{\delta} & 0 \\ 0 & \frac{1}{\beta} \emtx \Gamma
   = \bmtx \frac{1}{\delta} \Delta  & 0 \\
      0 & \frac{1}{\beta} M_{LTV} \emtx.
\end{align}
The scaling normalizes each block so that $\| \Gamma \| \le 1$.

Next, the uncertainty is LTI and hence $d_1\Delta = \Delta d_1$ for any scalar $d_1>0$. (In fact, this relation holds even if $d_1$ is also an LTI system but we will not pursue this generalization.) Similarly, the learning dynamics are also linear and hence $d_2M_{LTV} = M_{LTV} d_2$ for any scalar $d_2>0$. It follows that the normalized systems can be equivalently written, for any $d_1,d_2>0$, as:
\begin{align}
\tilde{\Gamma} := 
   \bmtx \frac{1}{d_1\delta} & 0 \\ 0 & \frac{1}{d_2 \beta} \emtx \Gamma
   \bmtx d_1 & 0 \\ 0 & d_2 \emtx.
\end{align}
This discussion leads to the following scaled small gain result.
\begin{theorem}
\label{thm:scaledSG}
Consider the interconnection $F_U(P,\Gamma)$ 
where $P:\ell_{pe} \to \ell_{pe}$ and $\Gamma:\ell_{pe} \to \ell_{pe}$ are linear systems with finite induced $\ell_p$-norm. Assume $\Gamma:= \bsmtx \Delta & 0 \\ 0& M_{LTV} \esmtx$ where
$\|\Delta\| \le \delta$ and $\|M_{LTV}\| \le \beta$.
Partition $P$ as:
\begin{align}
\bmtx \bar{p} \\ x \emtx
=
\bmtx P_{11} & P_{12} \\
P_{21} & P_{21} \emtx
\,
\bmtx \bar{q} \\ d \emtx,
\end{align}
where $\bar{p}:= \bsmtx p \\ \hat{w} \esmtx$ and  $\bar{q}:= \bsmtx q \\ u^{oco} \esmtx$ are the inputs and outputs of $\Gamma$. The interconnection has finite induced $\ell_p$-norm, i.e. $\|F_U(P,\Gamma)\| < \infty$, 
if there exists scalars $d_1, d_2>0$ such that 
\begin{align}
\label{eq:P11til}
\tilde{P}_{11}:=\bmtx \frac{1}{d_1} \, I & 0 \\ 0 & \frac{1}{d_2} \, I \emtx
P_{11}
\bmtx d_1\delta\, I & 0 \\ 0 & d_2\beta\, I \emtx
\end{align}
satisfies $\| \tilde{P}_{11}\| < 1$.
\end{theorem}
\begin{proof}
Define a scaled version of the nominal dynamics $P$ as:
\begin{align*}
    \tilde{P} = 
\left[\begin{array}{cc|c}
     \frac{1}{d_1} I & 0 & 0 \\
     0 & \frac{1}{d_2} I & 0 \\ \hline
     0 & 0 & I
\end{array} \right]
\bmtx P_{11} & P_{12} \\ P_{21} & P_{22} \emtx
\left[\begin{array}{cc|c}
     d_1\delta \, I & 0 & 0 \\
     0 & d_2\beta \, I & 0 \\ \hline
     0 & 0 & I
\end{array} \right].
\end{align*}
The constants introduced in the scaled plant $\tilde{P}$ cancel those introduced for $\tilde{\Gamma}$ in \eqref{eq:GamTil}. In other words, $F_U(P,\Gamma)$ and $F_U(\tilde{P},\tilde{\Gamma})$ define the same dynamics from $d$ to $x$. Moreover, $\|\tilde{P}_{11}\| <1$ and $\|\Gamma\|\le 1$ by assumption.  It follows from the small gain theorem (Theorem~\ref{thm:smallgain}) that $F_U(\tilde{P},\tilde{\Gamma}) = F_U(P,\Gamma)$ has finite induced $\ell_p$-norm.
\end{proof}

The scalings $d_1$ and $d_2$ in the robust stability condition (Theorem~\ref{thm:scaledSG}) can be used to reduce the conservatism of the small gain condition (Theorem~\ref{thm:smallgain}). They are known as $D$-scales in the robust control literature (\cite{packard93}
and Chapter 11 in \cite{zhou95}) and are used in structured singular value robust stability tests.

\subsection{Bounding the LTV Dynamics}
\label{sec:boundingMLTV}

In this section, we provide a result specific to the induced $\ell_\infty$-norm for the OCO control implementation. The induced $\ell_\infty$-norm is useful as it allows us to relate $\|M_{LTV}\|_{\infty\to\infty}$ to $\|M_t\|_{\infty\to\infty}$. The robust stability constraint can then be imposed as a point-wise in time constraint $\beta$ on the coefficients $\|M_t\|_{\infty\to\infty}\le\beta$ in the projection step of OPGD. We discuss this further in Section~\ref{sec:OPGD} and~\ref{sec:robust-oco}.

The dynamics $M_{LTV}$ in~\eqref{eq:mltv} can be expressed as:
\begin{align}
\label{eq:Msys2}
    u_t^{oco} &= M_t \hat{W}_t,
\end{align}
where
\begin{align}
\label{eq:stackedM}
    M_t & := \bmtx M_t^{[0]} & \cdots
    & M_t^{[H-1]} \emtx
    \in \R^{n_u \times n_xH},
    \mbox{ and } \\
\label{eq:stackedW}
    \hat{W}_t & := \bsmtx \hat{w}_{t} \\ \vdots \\ \hat{w}_{t-H+1} \esmtx \in \R^{n_x H}
\end{align}
are the stacked FIR coefficients and estimated disturbance history. The following theorem relates the induced $\ell_\infty$-norm of the system $M_{LTV}$ to the matrix induced $\infty$-norm of $M_t$.

\begin{theorem}
\label{thm:Mltv-bound}
Let $M_{LTV}$ be the LTV system defined in \eqref{eq:Msys2} and $M_t$ be the stacked gains defined in \eqref{eq:stackedM}. Then
\begin{align}
  \label{eq:MltvEquality}
   \| M_{LTV} \|_{\infty \to \infty} = \sup_t \| M_t\|_{\infty \to \infty}.
\end{align}
\end{theorem}

\begin{proof}
The equality in \eqref{eq:MltvEquality} is shown in two steps: (A) $\| M_{LTV} \|_{\infty \to \infty} \le \sup_t \| M_t \|_{\infty \to \infty}$
and (B) $\| M_{LTV} \|_{\infty \to \infty} \ge \sup_t \| M_t \|_{\infty \to \infty}$.

First, we show direction (A). Let $\hat{w}$ and $u^{oco}$ be any input-output pair of $M_{LTV}$.
Equation (\ref{eq:Msys2}) and the definition of the induced matrix norm imply that
\begin{align}
  \| u^{oco} \|_\infty &= \sup_t \| M_t \hat{W}_t \|_\infty \nonumber \\
  &\le \sup_t \| M_t \|_{\infty \to \infty} \cdot \| \hat{w} \|_\infty.
\end{align}
Thus,
$\frac{\| u^{oco} \|_\infty }{ \| \hat{w} \|_\infty }\le \sup_t \, \| M_t \|_{\infty\to\infty}$
so that
$\| M_{LTV} \|_{\infty\to\infty} \le \sup_t \| M_t\|_{\infty\to\infty}$.
Hence, claim (A) holds.

Next, we show direction (B).  Suppose
$\sup_t \| M_t \|_{\infty \to \infty}$
achieves its maximum at some finite time $t_0$.  (The proof can be modified if the supremum occurs as $t\to \infty$.)  Then there exists a vector $\hat{W}_0$ such that
\begin{align*}
    \|\hat{W}_0\|_\infty = 1 \mbox{ and } \|M_{t_0} \, \hat{W}_0\|_\infty = \sup_t \| M_t \|_{\infty \to \infty}.
\end{align*}
We can use the vector $\hat{W}_0$ to construct a signal $\hat{w}_0$ such that $u^{oco}=M_{LTV} \, \hat{w}_0$ satisfies
\begin{align*}
    \| u^{oco}\|_\infty \ge \left [ \sup_t \| M_t \|_{\infty \to \infty}  \right]  \|\hat{w}_0\|_\infty.
\end{align*}
Hence, claim (B) holds.
\end{proof}

\section{Application to OCO}
\label{sec:application}

In this section, we demonstrate how the main results can be applied to ensure robust stability of existing OCO controllers. 
We focus on the OCO controllers in \cite{agarwal19,agarwal19ANIPS}
where the coefficients of $M_{LTV}$ are updated via 
OPGD.

\subsection{Estimator Design}
\label{sec:agarwal-estimator}

The class of OCO controllers defined by \cite{agarwal19} considers the feedback system with OCO control (Figure \ref{fig:oco-control}) and no uncertainty (Figure \ref{fig:uncertainty}) when $\Delta(z)=0$. In this case, a perfect plant model is assumed $\tilde{G}(z)=G(z)$. Thus, the nominal plant dynamics can be used to design an estimator $E(z)$ and OPGD to update the coefficients in $M_{LTV}$. Later, we will show how the OPGD projection step can be modified to ensure robust stability for the case that there is uncertainty $\Delta(z)\ne0$.

Without uncertainty, the plant dynamics with unknown disturbance reduce to:
\begin{align*}
    x_{t+1} = Ax_t + Bu_t + Bd_t.
\end{align*}
Note that $Bd_t$ is the effective disturbance on the state at time $t$. Assuming the state $x_t$ is measurable, we can perfectly reconstruct this effective disturbance at the previous time step. Use the measured state and rearranging the plant dynamics:

\begin{align}
\label{eq:what}
    \hat{w}_t = x_t - Ax_{t-1} - Bu_{t-1}.
\end{align}
With no uncertainty, this estimator perfectly reconstructs the effective disturbance with a one-step delay: $\hat{w}_t=Bd_{t-1}$. However, perfect reconstruction is no longer guaranteed with uncertainty, i.e. if $\Delta(z)\ne 0$ then $\hat{w}_t\ne Bd_{t-1}$. In this case, $\hat{w}_t$ is considered an estimate of $Bd_{t-1}$.

The disturbance reconstruction \eqref{eq:what} can be expressed in state-space form as:
\begin{align*}
    x^e_{t+1} 
    & = 0 \,  x^e_t - A \, x_t - B u_t  \\
    \hat{w}_t
    &=
    x^e_t   + x_t,
\end{align*}
where $x^e_t = -A x_{t-1}- B u_{t-1}$ is the estimator state. This has the form of the general LTI estimator $E(z)$ in \eqref{eq:estimator}.
The estimates $\hat{w}_t$ of past disturbances are used to update the FIR coefficients $M_t$ defined in~\eqref{eq:stackedM} by minimizing an “ideal” cost which we describe next.

\subsection{OPGD on an Ideal Cost}
\label{sec:OPGD}

The coefficients $M_t$ are updated at each time step via OPGD in the direction of an “ideal” (per-step) cost. This cost is associated with the nominal plant dynamics (\ref{eq:nominal-plant}) and a per-step cost function. Here, we consider quadratic per-step costs:
\begin{align}
\label{eq:per-step}
	c(x_t,d_t) = x_t^\top Q \, x_t + u_t^\top R \, u_t,
\end{align}
where $Q=Q^\top\succeq0\in\R^{n_x \times n_x}$ and $R=R^\top\succ0\in\R^{n_u \times n_u}$. Note that the finite-horizon cost is defined as:
\begin{align}
\label{eq:total-cost}
    J_T(x,d) = \sum_{t=0}^T c(x_t,d_t),
\end{align}
where $T$ is the total time horizon. The ideal cost $g(M)$ is defined for any static gain $M \subset \mathbb{R}^{n_u \times n_xH}$ based on this per-step cost (\ref{eq:per-step}) which is computed and defined as follows.

Let $\tilde x_\tau\in\R^{n_x}$ and $\tilde u_\tau\in\R^{n_u}$ denote the ideal state and control input at time $\tau$, respectively. The ideal state and input are initialized at $\tau=t-H$ by:
\begin{align}
\label{eq:ideal-ic}
	\tilde x_{t-H}=0
 \mbox{ and }
	\tilde u_{t-H} = \sum_{i=0}^{H-1} M^{[i-1]} \, w_{t-H-i}.
\end{align}
where $t$ is the current time. The ideal state and control input are then computed for $\tau = t-H+1,\ldots,t$ by iterating over the plant dynamics with the static gains $M$:
\begin{align}
\label{eq:xideal}
	\tilde x_\tau & = A \, \tilde x_{\tau-1} + B \, \tilde u_{\tau-1} + \hat{w}_{\tau-1} \\
\label{eq:uideal}
	\tilde u_\tau & = -K \, \tilde x_\tau + \sum_{i=0}^{H-1} M^{[i]} \, \hat{w}_{\tau-i}.
\end{align}
The ideal cost is then defined as $g(M) := c( \tilde{x}_t, \tilde{u}_t)$. In other words, the ideal cost $g(M)$ is the cost of the plant dynamics evolving with static gain $M$ over the learning horizon $H$, neglecting dynamics beyond time $t-H$. The coefficients $M_t$ are updated via OPGD on this ideal cost:
\begin{align}
\label{eq:projection}
	M_{t+1} = \Pi_\mathcal{M} \left( M_t-\eta \nabla_M g(M_t) \right),
\end{align}
where $\eta$ is the learning rate, and $\Pi_\mathcal{M}$ is the projection of the gradient step of $M_t$ onto a constraint set $\mathcal{M}$. Additional details are given in \cite{agarwal19,agarwal19ANIPS}. Next, we show how the constraint set $\mathcal{M}$ can be modified to ensure the robust stability of the OCO feedback system (Figures \ref{fig:uncertainty} and \ref{fig:oco-control}) when $\Delta(z)\ne0$.

\subsection{Robust OCO Control}
\label{sec:robust-oco}

Assuming the uncertainty $\Delta(z)$ is bounded by some $\delta$, i.e., $\Vert \Delta \Vert_{\infty\to\infty}\le\delta$, we can use a bisection to find the required bound $\beta$ on the FIR filter $\Vert M_{LTV} \Vert_{\infty\to\infty} \le \beta$ such that the robust stability condition (Theorem~\ref{thm:scaledSG}) is satisfied. Larger values of $\beta$ risk stability, yet can improve disturbance rejection as they allow the OCO more freedom to adapt to the gains in $M_{LTV}$. Thus, it is important to determine the largest possible value of $\beta$ such that the robust stability condition holds. We refer to this $\beta$ as the stability bound. Theorem~\ref{thm:Mltv-bound} allows us to impose this constraint as a point-wise in time constraint on the FIR coefficients $M_t$.

Once the constraint $\beta$ has been determined, we can impose the constraint by defining the constraint set $\mathcal{M}$ as:
\begin{align}
\label{eq:constraint-set}
    \mathcal{M} := \Big\{ M \in \R^{n_u \times n_x H}
    :
    \Vert M \Vert_{\infty\to\infty} \le \beta \Big\}.
\end{align}
Thus, the projection $\Pi_\mathcal{M}$ can be implemented by:
\begin{align}
\label{eq:projection-explicit}
	M_{t+1} = 
	\begin{cases}
	M_{step}, & \Vert M_{step} \Vert_{\infty\to\infty} \le \beta \\
	\beta\left(\frac{M_{step}}{\Vert M_{step} \Vert_{\infty\to\infty}}\right), & \Vert M_{step} \Vert_{\infty\to\infty} > \beta,
    \end{cases}
\end{align}
where $M_{step}:=M_t - \eta\nabla_Mg(M_t)$ is the gradient step of the FIR coefficients $M_t$, defined at time $t$. The constraint set $\mathcal{M}$ defined in \eqref{eq:constraint-set} and projection $\Pi_\mathcal{M}$ in \eqref{eq:projection-explicit} can be implemented as part of Algorithm 1 in \cite{agarwal19}. The numerical results in the following section are based on this implementation.




\section{Numerical Results}
\label{sec:numerical}

In this section, we provide numerical results of OPGD on a plant with uncertainty. Although we do not explicitly use the robust stability condition (Theorem~\ref{thm:scaledSG}) to compute the stability bound $\beta$, we perform numerical studies to illustrate its effect. Future studies will focus on computing the exact bound, while the results here suggest that a stability bound $\beta$ exists.


Here, unconstrained OCO (U-OCO) refers to OCO control with $\beta=\infty$, i.e., unconstrained FIR filter gains $M_t$. Constrained OCO (C-OCO) refers to OCO control with gains $M_t$ bounded by some $\beta<\infty$. We compare results between U-OCO and C-OCO on the following models:
\begin{align*}
    G(z) &= \frac{0.1}{z-0.9} \\
    F(z) &= \frac{0.1185z+0.1145}{z^2-1.672z+0.9048} 
\end{align*}
where $G(z)$ and $F(z)$ are the nominal plant and unmodeled high frequency actuator dynamics, respectively. Note that $\Delta(z)=F(z)-1$ and $\tilde{G}(z)=G(z)F(z)$. The following disturbance $d_t $ was generated to perturb the control input $u_t$:
\begin{align*}
    d_t = \begin{cases}
    100 & 0 \le t \le 500 \\
    -100 & 500 < t \le T,
    \end{cases}
\end{align*}
where the time horizon is $T=1000$. We use the quadratic per-step cost $c(x_t,d_t)$ and total cost $J_T(x,d)$ defined in~\eqref{eq:per-step} and~\eqref{eq:total-cost}, respectively, with $Q=1$ and $R=10^{-1}$. The state-feedback gain $K=0.15$, learning horizon $H=1$, and learning rate $\eta=5\times 10^{-4}$ are used for all simulations. Note that $K=0.15$ is stabilizing for both the nominal and true plant dynamics.

Figure~\ref{fig:u-oco} shows the per-step cost $c(x_t,d_t)$ and estimated disturbance $\hat{w}_t$ of U-OCO at each time $t$. We compare the performance with a perfect (red dashed) and imperfect (blue solid) plant model. Again, a perfect model is without uncertainty $\Delta(z)=0$, and an imperfect model is with uncertainty $\Delta(z)\ne0$. The disturbance is perfectly reconstructed $\hat{w}_t=Bd_{t-1}$(see Section~\ref{sec:agarwal-estimator}) with a perfect plant model. However, with an imperfect plant model, this is not the case $\hat{w}_t\ne Bd_{t-1}$. Since the ideal cost $g(M)$ computation assumes a perfect plant model and disturbance estimates, this mismatch introduces an error in the coefficient update $M_{t+1}$. This causes an instability which is reflected by the per-step cost and estimated disturbance growing unbounded. On the other hand, U-OCO performance is stable without uncertainty because the disturbance is estimated perfectly. Thus, the constraint $\beta$ is needed on the coefficient update to ensure stability for the imperfect plant model.

\begin{figure}[h!t]
    \centering
    \includegraphics[width=9cm]{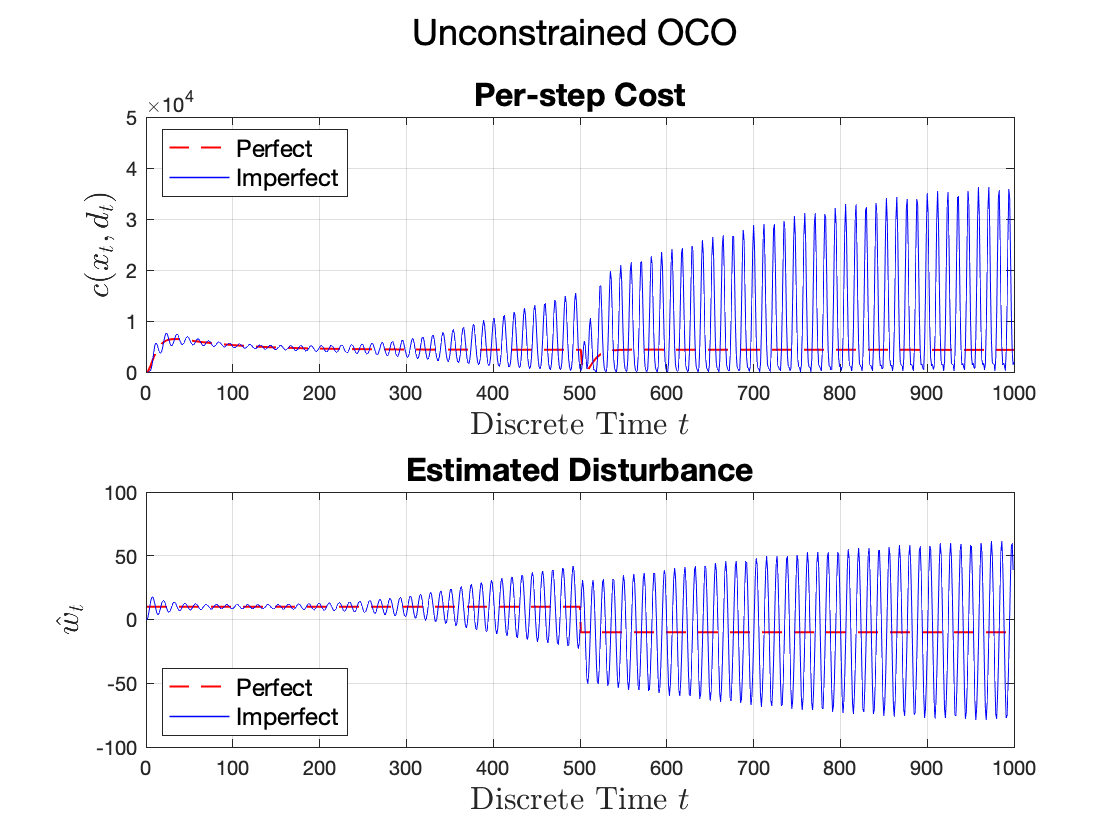}
    \caption{Per-step cost (top) and disturbance estimate (bottom) of running U-OCO on a perfect (red dashed) and imperfect (blue solid) plant model. U-OCO is stable with a perfect model and unstable for an imperfect model.}
    \label{fig:u-oco}
\end{figure}

Figure~\ref{fig:c-oco} shows the per-step cost $c(x_t,d_t)$ and estimated disturbance $\hat{w}_t$ of C-OCO for $\beta=1.5$ at each time $t$. Again, we compare the performance with a perfect (red dashed) and imperfect (blue solid) plant model. As mentioned before, an error in the disturbance estimate introduces an error in the ideal cost gradient. The ideal cost gradient error can cause the gradient step $M_{step}=M_t-\nabla_Mg(M_t)$ to grow too large in the wrong direction. When the constraint $\beta$ is chosen such that the robust stability condition (Theorem~\ref{thm:scaledSG}) is satisfied, the effect of uncertainty induced error on the gradient step of the coefficient update is limited. This is illustrated in Figure~\ref{fig:c-oco} as the performance of C-OCO on the imperfect plant model eventually recovers the performance on the perfect model with $\beta=1.5$. Thus, imposing the constraint $\beta$ can ensure that OCO is robust to uncertainty.

As mentioned in Section~\ref{sec:robust-oco}, the choice of $\beta$ is critical. Figure~\ref{fig:beta-sweep} shows the averaged per-step cost $J_T(x,d)/T$ for C-OCO as a function of $\beta$. Again, we compare the performance with a perfect (red dashed) and imperfect (blue solid) plant model. When $\beta=0$, the OCO has no freedom to learn the disturbance, and pure state-feedback (SF) is recovered for both the perfect and imperfect plants (red and blue circles, respectively). As $\beta$ is increased, the OCO is allowed more freedom to learn the disturbance, and we see similar improved performance in both the perfect and imperfect plants. However, when $\beta$ is "too large" such that the robust stability condition (Theorem~\ref{thm:scaledSG}) no longer holds, C-OCO on the imperfect plant becomes unstable. Figure~\ref{fig:beta-sweep} suggests that the stability bound occurs around $\beta=1.5$. Note that once the constraint $\beta$ becomes inactive, C-OCO recovers U-OCO performance for the perfect and imperfect plants (red and blue squares, respectively). For the perfect plant, this indicates a limit as to how much the OCO can improve on the baseline controller. For the imperfect plant, this indicates a limit as to how much the OCO performance can be degraded by uncertainty. Hence, there is this trade off between OCO performance and robustness to uncertainty.

\begin{figure}[h!t]
    \centering
    \includegraphics[width=9cm]{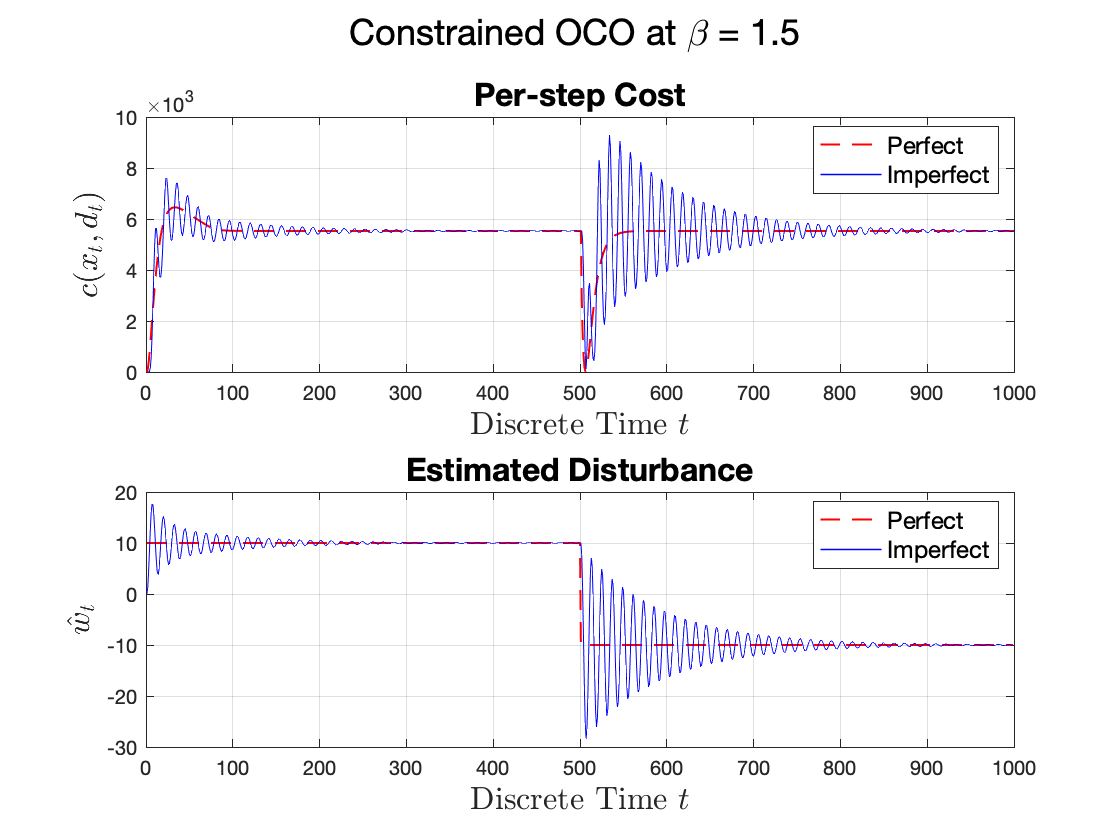}
    \caption{Per-step cost (top) and disturbance estimate (bottom) of running C-OCO at $\beta=1.5$ on a perfect (red dashed) and imperfect (blue solid) plant model. C-OCO is stable for the perfect and imperfect models.}
    \label{fig:c-oco}
\end{figure}

\begin{figure}[h!t]
    \centering
    \includegraphics[width=9cm]{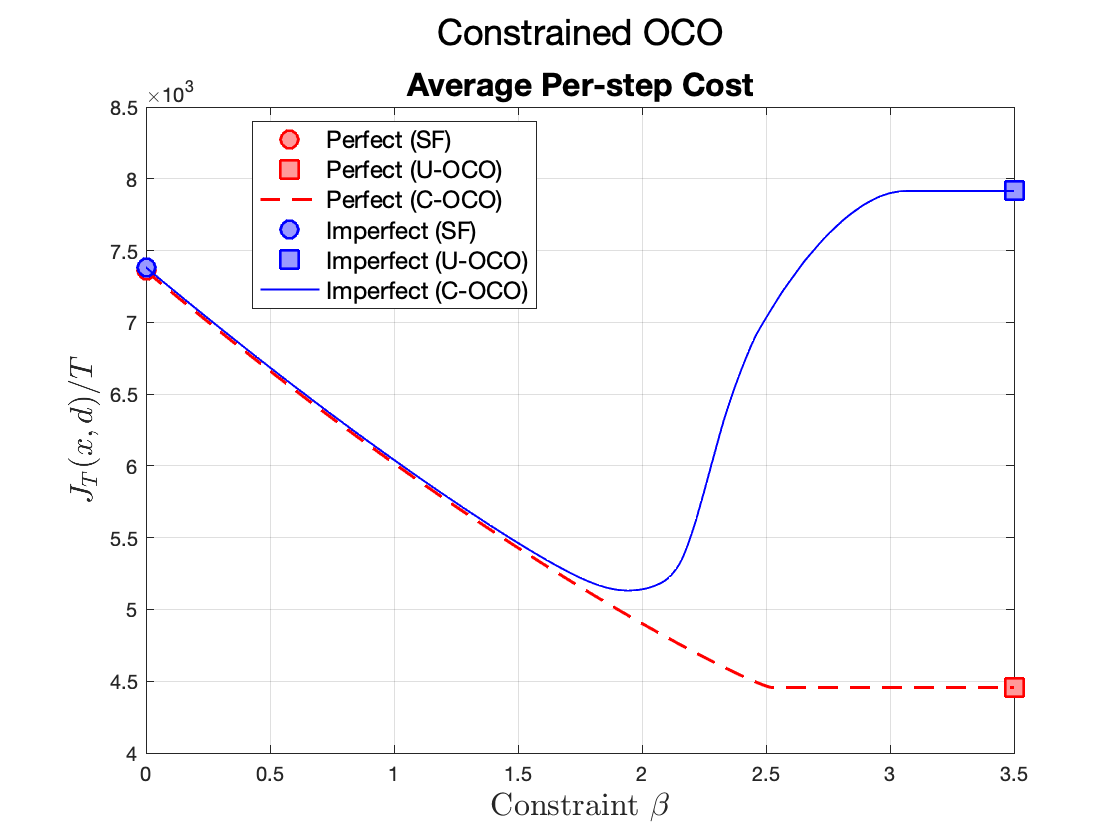}
    \caption{Averaged per-step cost of running C-OCO for varying $\beta$ on a perfect (red dashed) and imperfect (blue solid) plant model. C-OCO results in improved performance for the perfect and imperfect models until the constraint becomes too large and C-OCO on the imperfect model becomes unstable.}
    \label{fig:beta-sweep}
\end{figure}

\section{Conclusion}
\label{sec:conclusion}

In this paper, we establish a robust stability condition using the small gain theorem for a class of OCO controllers with memory and use this result to develop an OCO control algorithm (C-OCO) robust to model uncertainty. In particular, we impose this constraint on the controller by bounding the LTV dynamics of the OCO controller point-wise in time. We provide numerical results to illustrate that imposing the robust stability constraint keeps the closed-loop system stable when it would go unstable otherwise. Future work will study the numerical implementation of the scaled small gain theorem to compute the stability bound $\beta$.



\bibliography{Bibilography/RegretControl,Bibilography/RobControl,Bibilography/OCO}
\bibliographystyle{ieeetr}

\vspace{12pt}

\end{document}